\documentclass{svmult}

\usepackage[utf8]{inputenc}
\usepackage[english]{babel}

\usepackage{mathptmx}
\usepackage{helvet}
\usepackage{courier}        
\usepackage[bottom]{footmisc}

\usepackage{amssymb,amsmath}

\newcommand{\im}{\I}

\newcommand{\N}{\mathbb{N}}
\newcommand{\R}{\mathbb{R}}
\newcommand{\Z}{\mathbb{Z}}
\newcommand{\bC}{\mathbb{C}}
\newcommand{\defeq}{:=}
\newcommand{\tprod}{\otimes}
\newcommand{\one}{\mathbf{1}}
\newcommand{\cH}{\mathcal{H}}
\newcommand{\cA}{\mathcal{A}}
\newcommand{\trk}{\mathrm{tr}}
\newcommand{\ctprod}{\hat{\otimes}}
\newcommand{\ds}{\circ}
\newcommand{\fa}{\dagger}
\newcommand{\cF}{\mathcal{F}}
\newcommand{\cAn}{\mathcal{A}_0}
\newcommand{\cAe}{\mathcal{A}_{\mathrm{e}}}
\newcommand{\tcl}{\mathcal{T}}
\newcommand{\la}{*}
\newcommand{\fg}{\mathfrak{g}}
\newcommand{\fge}{\mathfrak{g}_{\mathrm{e}}}
\newcommand{\fh}{\mathfrak{h}}
\newcommand{\fm}{\mathfrak{m}}

\newcommand{\fn}{\mathfrak{n}}
\newcommand{\asl}{\Xi}
\newcommand{\lop}{\|}
\newcommand{\rop}{\|_{\mathrm{op}}}
\newcommand{\vac}{\psi_0}
\newcommand{\fp}{\mathfrak{p}}
\newcommand{\hol}{\mathrm{Hol}}
\newcommand{\clin}{\mathcal{B}}
\newcommand{\sig}[1]{[#1]}
\newcommand{\fdg}[1]{|#1|}

\newtheorem{dfn}{Definition}[section]

\newtheorem{lem}[dfn]{Lemma}
\newtheorem{prop}[dfn]{Proposition}
\newtheorem{thm}[dfn]{Theorem}

\smartqed
\newcommand{\myqed}{\qed}

\begin{document}

\title*{Coherent states in fermionic Fock-Krein spaces and their amplitudes}
\author{Robert Oeckl}
\institute{Robert Oeckl \at Centro de Ciencias Matemáticas,
Universidad Nacional Autónoma de México,
C.P.~58190, Morelia, Michoacán, Mexico, \email{robert@matmor.unam.mx}}

\maketitle

\abstract{%
We generalize the fermionic coherent states to the case of Fock-Krein spaces, i.e., Fock spaces with an idefinite inner product of Krein type. This allows for their application in topological or functorial quantum field theory and more specifically in general boundary quantum field theory. In this context we derive a universal formula for the amplitude of a coherent state in linear field theory on an arbitrary manifold with boundary.
}

\section{Introduction}

A key property of bosonic coherent states is that they can closely mimic the behavior of classical states in the classical counterpart of a quantum system. This is strikingly exemplified in the following situation. Consider a classical field theory, linear for simplicity, determined by partial differential equations of motion in spacetime. Consider a spacetime region $M$ and denote by $L_{\partial M}$ the real vector space of data on the boundary $\partial M$ of $M$, i.e., the space of germs of solutions defined in a neighborhood of $\partial M$. The equations of motions in $M$ induce a subspace $L_M\subseteq L_{\partial M}$ of those boundary data that admit a continuation to the interior, i.e., are classically allowed in $M$.

A quantization of the theory provides a complex structure on $L_{\partial M}$ and an inner product, making it into a complex Hilbert space. The state space of the quantum theory associated to the boundary is then the Fock space $\cF(L_{\partial M})$ over $L_{\partial M}$. The usual normalized coherent states can be described through a map $\tilde{K}:L_{\partial M}\to \cF(L_{\partial M})$. That is, for any initial data $\xi$ there is a corresponding normalized coherent state $\tilde{K}(\xi)$ in the Fock space. The quantum amplitude $\rho_M$ for such a coherent state in the region $M$ is given by the following remarkably simple formula \cite{Oe:holomorphic},
\begin{equation}
  \rho_M\left(\tilde{K}(\xi)\right)=\exp\left(-\im\,\omega_{\partial M}(\xi^{\text{R}},\xi^{\text{I}})-\frac{1}{2} g_{\partial M}(\xi^{\text{I}},\xi^{\text{I}})\right) .
  \label{eq:bosampl}
\end{equation}
Here $g_{\partial M}$ denotes the real part of the inner product on $L_{\partial M}$ and $\im\,\omega_{\partial M}$ the imaginary part. $\xi=\xi^{\text{R}}+\xi^{\text{I}}$ is the decomposition of $\xi$ into a component $\xi^{\text{R}}\in L_{\partial M}$ and a component $\xi^{\text{I}}\in L_{\partial M}^\perp$, where $L_{\partial M}^\perp$ is the complement of $L_{M}$ in $L_{\partial M}$ with respect to the real part of the inner product.

What is the physical interpretation of formula (\ref{eq:bosampl})? Suppose first that we consider data on the boundary that is classically allowed in the interior, i.e., data $\xi\in L_{M}$. Then, $\xi^{\text{R}}=\xi$ and $\xi^{\text{I}}=0$, the expression inside the exponential vanishes, and the amplitude is unity. Now, we switch on a classically forbidden component $\xi^{\text{I}}\in L_M^\perp$. The first term inside the exponential is imaginary and turns on an irrelevant phase. The second term, however, is real and negative (as the inner product is positive definite), leading to an exponential suppression of the amplitude. This is precisely what one should expect in the quantum theory from a classically forbidden configuration. Viewed the other way round, the amplitude formula precisely provides us here with a semiclassical interpretation of the coherent states.

In the case of fermionic systems a semiclassical interpretation of coherent states is very little explored, not least because the very notion of a classical fermionic system appears to be problematic. In view of the bosonic example it is not unreasonable to expect that an explicit expression for the amplitude of a fermionic coherent state might lead to some insight here. Such a formula is the principal result of the present work. To get there we first need to generalize the notion of fermionic coherent state from the setting of Fock spaces to that of Fock-Krein spaces. This is due to the finding that fermionic state spaces associated to general hypersurfaces in quantum field theory are necessarily indefinite inner product spaces rather than Hilbert spaces \cite{Oe:freefermi}. More specifically, they turn out to be Krein spaces, i.e., complete direct sums of a positive definite and a negative definite part. Only in the (text book) special case of spacelike hypersurface can one consistently restrict to Hilbert spaces. Thus in Section~\ref{sec:krein} we briefly review relevant facts about Krein spaces and in Section~\ref{sec:fockkrein} the construction of Fock spaces based on them.

The coherent states are constructed following the group theoretic approach of Gilmore \cite{ZFG:coherent} and Perelomov \cite{Per:coherent}, generalizing the treatment in \cite{Oe:fermicoh}. This implies that we consider a \emph{dynamical Lie algebra} acting on Fock-Krein space, as reviewed and generalized in Section~\ref{sec:dynlalg}. The coherent states arise then through the exponentiation of a certain Lie subalgebra, see Section~\ref{sec:cohstates}. It turns out that also here the results from the Fock space setting generalize mostly straightforwardly to the Fock-Krein space setting. In particular, the fermionic Fock-Krein space is a reproducing kernel Krein space (Theorem~\ref{thm:repkkrein}). Section~\ref{sec:amplitude} contains the main results of the present work. The principal result is the amplitude formula, Theorem~\ref{thm:amplcoh}. The proof involves, among other steps, a combinatorial result in the spirit of Pólya's theory of counting. We close the section with a derivation of the inner product formula for coherent states (Theorem~\ref{thm:cohip}) from the amplitude formula. We use the framework of (fermionic) general boundary quantum field theory \cite{Oe:freefermi}, whose axioms are recalled in an appendix.

\section{Krein spaces}
\label{sec:krein}

We recall in this section essentials of the notion of Krein space, see e.g.\ \cite{Bog:indipspaces}.

\begin{dfn}
  Let $V$ be a real (or complex) topological vector space with a real (or complex) non-degenerate inner product. Assume there are orthogonal subspaces $V^+$ and $V^-$ generating $V$ such that the inner product is positive definite on $V^+$ and negative definite on $V^-$. Moreover, suppose that $V^+$ is complete with respect to the inner product and that $V^-$ is complete with respect to the negative of the inner product. Then, $V$ is called a \emph{Krein space}. We call the presentation of $V$ as the direct sum $V=V^+\oplus V^-$ a \emph{decomposition} of $V$, with $V^+$ called \emph{positive part}, and $V^-$ called \emph{negative part}. We call a Krein space together with a fixed decomposition a \emph{strict Krein space}.
\end{dfn}

Given a Krein space $V$, a decomposition induces a $\Z_2$-grading for which we use the notation
\begin{equation}
\sig{v}\defeq\begin{cases}0 & \text{if}\; v\in V^+\\
1 & \text{if}\; v\in V^-\end{cases} .
\end{equation}

\begin{dfn}
Given a Krein space $V$ with a decomposition, an \emph{adapted orthonormal basis} is an orthonormal basis of the positive part in the Hilbert space sense together with an orthonormal basis of the negative part in the Hilbert space sense for the negative of the inner product. An \emph{orthonormal basis} of $V$ is an orthonormal basis adapted to some decomposition.
\end{dfn}

We restrict in the following to \emph{separable} Krein spaces, that is, Krein spaces that admit a countable orthonormal basis. We denote the inner product of $V$ by $\{\cdot,\cdot,\}$.
Given an orthonormal basis $\{\zeta_i\}_{i\in I}$ we have the completeness relation,
\begin{equation}
  \{v,w\}=\sum_{i\in I} \{\zeta_i,\zeta_i\} \{v,\zeta_i\} \{\zeta_i,w\},\qquad\forall v,w\in V .
  \label{eq:cness}
\end{equation}
Note that this relation holds for any orthonormal basis adapted to any decomposition.

Given a decomposition a Krein space $V$ becomes a Hilbert space by changing the sign of the inner product on the negative part. This induces a norm on $V$. While the norms for different decompositions are different they are equivalent in generating the same topology. This norm in turn induces an \emph{operator norm} on the space $\clin(V)$ of continuous operators on $L$. Again, these norms are different for different decompositions, but generate the same topology. The notion of \emph{trace class operator} defined in the Hilbert space sense, turns out also to be independent of the decomposition. We denote the space of such operators by $\tcl(V)$.

The \emph{trace} of a trace class operator $\lambda\in\tcl(V)$ is given by,
\begin{equation}
 \trk(\lambda)\defeq\sum_{i\in I} \{\zeta_i,\zeta_i\} \{\zeta_i, \lambda\zeta_i\},
\end{equation}
where $\{\zeta_i\}_{i\in I}$ is an orthonormal basis of $V$. Note that this expression is invariant not only under choice of orthonormal basis adapted to a fixed decomposition, but also under choice of decomposition. This can easily be seen with the completeness property (\ref{eq:cness}).

We define the adjoint $b^\la$ of an operator $b\in\clin(V)$ with respect to the Krein space inner product,
\begin{equation}
  \{b^\la v, w\}=\{v,b w\} \quad\forall v,w\in V .
\end{equation}

Let $V$ be a real or complex Krein space.
Let $\Lambda:V\to V$ be a continuous real linear map. We say that $\Lambda$ is \emph{real anti-symmetric} if,
\begin{equation}
 \Re\{v,\Lambda w\}=-\Re\{w,\Lambda v\} \quad\forall v,w\in V .
\end{equation}
We say that $\Lambda$ is a \emph{real isometry} if,
\begin{equation}
 \Re\{\Lambda v,\Lambda w\}=\Re\{v, w\} \quad\forall v,w\in V .
\end{equation}
We say that $\Lambda$ is a \emph{real anti-isometry} if,
\begin{equation}
 \Re\{\Lambda v,\Lambda w\}=-\Re\{v,w\} \quad\forall v,w\in V .
\end{equation}
We say that $\Lambda$ is \emph{involutive} if $\Lambda^2=\one_V$, where $\one_V$ denotes the identity operator on $V$. It is easy to see that $\Lambda$ is involutive and real anti-symmetric if and only if it is involutive and a real anti-isometry. Given a decomposition we say that a real isometry is \emph{adapted} if it preserves the decomposition, while a real anti-isometry is called adapted if it interchanges the decomposition.

Let $V$ be a complex Krein space.
If $\Lambda$ is real anti-symmetric and complex conjugate-linear we simply say that $\Lambda$ is \emph{conjugate anti-symmetric}.\footnote{In \cite{Oe:fermicoh} this property was simply called ``anti-symmetric''.} It is easy to see that this is equivalent to the property,
\begin{equation}
  \{v,\Lambda w\}=-\{w,\Lambda v\} \quad\forall v,w\in V .
  \label{eq:conjas}
\end{equation}
Then, $\Lambda^2$ is complex linear, self-adjoint, and negative in the Krein sense,
\begin{gather}
 \{v,\Lambda^2 w\}=-\{\Lambda w,\Lambda v\}=\{\Lambda^2 v, w\},\qquad\forall v,w\in V, \\
 \{v,\Lambda^2 v\}=-\{\Lambda v,\Lambda v\},\qquad\forall v\in V .
\end{gather}
Let $\asl(V)$ be the vector space of conjugate anti-symmetric maps $\Lambda:V\to V$ such that $\Lambda^2$ is trace class.
Even though the elements of $\asl(V)$ are not complex linear maps, the space $\asl(V)$ itself is naturally a complex vector space. This is because given $\Lambda\in\asl(V)$, the map $\im\Lambda$ defined by $(\im \Lambda)(v)\defeq \im\, \Lambda(v)$ is also in $\asl(V)$. But note, $\im\, \Lambda(v)=-\Lambda(\im v)$.

\section{Fock-Krein spaces}
\label{sec:fockkrein}

Let $L$ be a strict complex Krein space, i.e., a complex Krein space with a fixed decomposition. We denote the inner product on $L$ by $\{\cdot,\cdot\}$. We briefly recall the construction of the Fock space $\cF(L)$ as a strict complex Krein space based on $L$. We shall follow the notation and conventions of \cite{Oe:freefermi}. The inner product on $\cF(L)$ is denoted by $\langle\cdot,\cdot\rangle$. Thus, Fock space arises as the completion of a direct sum of strict Krein spaces,
\begin{equation}
 \cF(L)=\widehat{\bigoplus_{n=0}^\infty}\,\cF_n(L) .
\end{equation}
Here, $\cF_0(L)$ is the one-dimensional complex Hilbert space isomorphic to $\bC$. It carries the standard inner product $\langle\vac,\vac\rangle=1$, where $\vac$ is a vector that generates this space, also called the \emph{vacuum state}. Moreover, $\cF_n(L)$ is the space of continuous $n$-linear maps from $n$ copies of $L$ to $\bC$, anti-symmetric under transposition of any arguments. The inner product on $\cF_n(L)$ is given by,
\begin{equation}
  \langle \eta,\psi\rangle = 2^n n! \sum_{j_1,\ldots,j_n\in I}
  \{\zeta_1,\zeta_1\}\cdots \{\zeta_n,\zeta_n\}\,
  \overline{\eta(\zeta_1,\ldots,\zeta_n)}\psi(\zeta_1,\ldots,\zeta_n) ,
\end{equation}
where $\{\zeta_j\}_{j\in I}$ is an adapted orthonormal basis of $L$. (Compare formula (50) in \cite{Oe:locqft}. The relative factor arises from the difference between using a real and a complex basis.)
$\cF_n(L)$ is a strict Krein space with decomposition $\cF_n(L)=\cF_n^+(L)\oplus\cF_n^-(L)$ given by,
\begin{align}
 \cF_n^+(L) & \defeq \{\psi\in \cF_n: \psi(\xi_1,\dots,\xi_n)=0\quad\text{if}\quad\sig{\xi_1}+\cdots+\sig{\xi_n}\;\text{odd}\}, \\
 \cF_n^-(L) & \defeq \{\psi\in \cF_n: \psi(\xi_1,\dots,\xi_n)=0\quad\text{if}\quad\sig{\xi_1}+\cdots+\sig{\xi_n}\;\text{even}\} .
\end{align}
This induces a global decomposition as a strict Krein space $\cF(L)=\cF^+(L)\oplus\cF^-(L)$. On the other hand $\cF(L)$ is an $\N$-\emph{graded} space. This grading modulo two, i.e., the induced $\Z_2$-grading is also called \emph{fermionic grading} or \emph{f-grading} for short. For this grading we introduce the notation
\begin{equation*}
  \fdg{\psi}\defeq\begin{cases}0 & \text{if}\; \psi\in \cF_{2n}(L)\\
1 & \text{if}\; \psi\in \cF_{2n+1}(L) \end{cases} .
\end{equation*}
We shall use the term \emph{real f-graded isometry} for a map on an f-graded Krein space that is a real isometry on the even part and a real anti-isometry on the odd part.

For $\tau\in L$ we define the associated \emph{creation operator} $a_\tau^\fa$ and \emph{annihilation operator} $a_\tau$ as follows for $\psi\in \cF_n(L)$,
\begin{align}
  (a_{\tau}\psi)(\eta_1,\ldots,\eta_{n-1}) & =\sqrt{2}\, n\, \psi(\tau,\eta_1,\ldots,\eta_{n-1}) , \\
  (a_{\tau}^\fa\psi)(\eta_1,\ldots,\eta_{n+1})
  & =\frac{1}{\sqrt{2}\, (n+1)} \nonumber \\
  &\quad \sum_{k=1}^{n+1} (-1)^{k-1} \{\tau,\eta_k\}\,
   \psi(\eta_1,\ldots,\eta_{k-1},\eta_{k+1},\ldots,\eta_{n+1}) . \label{eq:actcr}
\end{align}
The adjoint is defined here with respect to the Krein space inner product,
\begin{equation}
  \langle a^\fa \psi,\psi'\rangle = \langle \psi,a \psi'\rangle\quad\forall \psi,\psi'\in \cF(L) .
\end{equation}

The CAR algebra $\cA$ is the unital algebra generated by the creation and annihilation operators and with the relations,
\begin{equation}
 a_{\xi+\tau}=a_\xi+a_\tau,\quad a_{\lambda\xi}=\lambda a_{\xi},\quad
 a_{\xi} a_{\tau} + a_{\tau} a_{\xi}=0,\quad
 a_{\xi}^\fa a_{\tau}+a_{\tau} a_{\xi}^\fa=\{\xi,\tau\}\one .
\label{eq:CAR}
\end{equation}
We denote by $\cAn$ the subalgebra that preserves degree and $\cAe$ the subalgebra that only preserves f-degree. In particular, $\cAn\subseteq\cAe\subseteq\cA$.

\section{Dynamical Lie algebra and its action}
\label{sec:dynlalg}

In this section we briefly recall the dynamical Lie algebra as defined in \cite{Oe:fermicoh} and perform the straightforward generalization from the Hilbert to the Krein space case.

For any $\lambda\in\tcl(L)$ define the operator $\hat{\lambda}:\cF(L)\to\cF(L)$ by,
\begin{equation}
  \hat{\lambda}\defeq \frac{1}{2} \sum_{i\in I} \{\zeta_i,\zeta_i\}
   \left(a_{\zeta_i}^{\fa} a_{\lambda(\zeta_i)} - a_{\lambda(\zeta_i)} a_{\zeta_i}^{\fa}\right)
   = \sum_{i\in I} \{\zeta_i,\zeta_i\}\, a_{\zeta_i}^{\fa} a_{\lambda(\zeta_i)}
     - \frac{1}{2}\trk(\lambda)\one .
\end{equation}
Here $\{\zeta_i\}_{i\in I}$ is an orthonormal basis of $L$. This is a Krein space generalization of the fermionic current operators \cite{HaKa:aqft}. The $*$-structures on $\tcl(L)$ and the induced operators are in correspondence, $\hat{\lambda^\la}=\hat{\lambda}^\fa$. The vector space $\fh^\bC$ of these current operators forms a complex Lie algebra, with the Lie bracket given by the commutator,
\begin{equation}
  [\hat{\lambda},\hat{\lambda'}]=\hat{\lambda''}, \qquad \text{with}\qquad \lambda''=\lambda'\lambda-\lambda\lambda' .
  \label{eq:hrel}
\end{equation}
Thus $\tcl(L)$ is naturally anti-isomorphic to $\fh^\bC$ as a complex Lie algebra.
$\fh^\bC$ has a natural real structure induced by its action on $\cF(L)$. Namely, the real Lie subalgebra $\fh$ is generated by the elements $\hat{\lambda}$ that are skew-adjoint, i.e., satisfy $\lambda=-\lambda^\la$. We call $\fh$ as the \emph{degree-preserving dynamical Lie algebra}. Its complexified universal enveloping algebra generates $\cAn$.

For any $\Lambda\in\asl(L)$ define the operator $\hat{\Lambda}:\cF(L)\to\cF(L)$ with adjoint $\hat{\Lambda}^\fa:\cF(L)\to\cF(L)$,
\begin{equation}
  \hat{\Lambda}\defeq \frac{1}{2} \sum_{i\in I} \{\zeta_i,\zeta_i\}
   a_{\zeta_i} a_{\Lambda(\zeta_i)} ,\qquad   \hat{\Lambda}^\fa= \frac{1}{2} \sum_{i\in I} \{\zeta_i,\zeta_i\}
   a_{\Lambda(\zeta_i)}^\fa a_{\zeta_i}^\fa  .
\end{equation}
For later use we also note explicit formulas for the action of these operators. For $\psi\in\cF_n(L)$ we have,
\begin{align}
  (\hat{\Lambda}\psi)(\eta_1,\ldots,\eta_{n-2}) & = n\, (n-1) \sum_{i\in I} \{\zeta_i,\zeta_i\} \psi(\Lambda(\zeta_i),\zeta_i,\eta_1,\ldots,\eta_{n-2}) ,\\
  (\hat{\Lambda}^\fa\psi)(\eta_1,\ldots,\eta_{n+2})
  & =\frac{1}{4\, (n+2)!} \nonumber \\
  & \quad \sum_{\sigma\in S^{n+2}} (-1)^{|\sigma|} \{\Lambda \eta_{\sigma(2)},\eta_{\sigma(1)}\}\,
   \psi(\eta_{\sigma(3)},\ldots,\eta_{\sigma(n+2)}) . \label{eq:actLhat}
\end{align}
Here $S^n$ denotes the symmetric group in $n$ elements, i.e., the group of permutation of $n$ elements. For a permutation $\sigma$ we define $|\sigma|=0$ if it is an even permutation, i.e., can be obtained as a composition of an even number of transpositions, and $|\sigma|=1$ otherwise.
Given a decomposition of the Krein space $L$ we have a corresponding operator norm on $\cF(L)$ that we denote $\lop\cdot\rop$. We note,
\begin{equation}
 \lop \hat{\Lambda}\rop^2=\lop \hat{\Lambda}^\fa\rop^2=\|\hat{\Lambda}^\fa\vac\|_{\cF(L)}^2=-\frac{1}{2}\trk(\Lambda_0^2)+\frac{1}{2}\trk(\Lambda_1^2) ,
\end{equation}
where $\Lambda=\Lambda_0+\Lambda_1$ and $\Lambda_0$ preserves positive and negative parts of the decomposition of $L$ while $\Lambda_1$ interchanges them.

With the commutator on $\cF(L)$ the operators $\hat{\Lambda}$ from a complex abelian Lie algebra $\fm_+$ and the operators $\hat{\Lambda}^\fa$ form a complex abelian Lie algebra $\fm_-$. (Note that the complex structures are induced from $\asl(L)$ for $\fm_+$ and its opposite for $\fm_-$.) We denote the complex vector space spanned by both $\fm^\bC=\fm_+\oplus\fm_-$. Its natural real subspace $\fm$ is spanned by the skew-adjoint operators $\hat{\Lambda}-\hat{\Lambda}^\fa$.
The direct sum $\fge\defeq \fh\oplus\fm$ turns out to be a real Lie algebra with complexification $\fge^\bC=\fh^\bC\oplus \fm^\bC$. In addition to (\ref{eq:hrel}) the Lie brackets are as follows (other brackets vanish),
\begin{align}
 & [\hat{\lambda},\hat{\Lambda}]=\hat{\Lambda'},
  && \text{with}\qquad \Lambda'=-\lambda\Lambda-\Lambda\lambda^\la, \nonumber\\
 & [\hat{\lambda},\hat{\Lambda}^\fa]=\hat{\Lambda'}^\fa,
  && \text{with}\qquad \Lambda'= \lambda^\la\Lambda+\Lambda\lambda, \nonumber\\
 & [\hat{\Lambda'},\hat{\Lambda}^\fa]=\hat{\lambda},
  && \text{with}\qquad \lambda=\Lambda'\Lambda .
  \label{eq:frel}
\end{align}
We shall refer to $\fge$ as the \emph{even dynamical Lie algebra}. Its complexified universal enveloping algebra generates $\cAe$.

We set $\hat{\xi}\defeq 1/\sqrt{2}\, a_{\xi}$ and consider the complex vector spaces spanned by the operators $\hat{\xi},\hat{\xi}^\fa$ denoting them by $\fn_+,\fn_-$ respectively. For a decomposition of $L$ we then have,
\begin{equation}
\lop\hat{\xi}\rop^2=\lop\hat{\xi}^\fa\rop^2=\|\hat{\xi}^\fa\vac\|_{\cF(L)}^2=\frac{1}{2}\|\xi\|_L^2 .
\end{equation}
$\fn_+$ is naturally isomorphic to $L$ as a complex vector space while $\fn_-$ is naturally isomorphic to $L$ with the opposite complex structure. Denote the direct sum by $\fn^\bC\defeq\fn_+\oplus\fn_-$ and its real subspace spanned by $\hat{\xi}-\hat{\xi}^\fa$ by $\fn$. Then, $\fg\defeq \fh\oplus\fm\oplus\fn$ forms a real Lie algebra and $\fg^\bC= \fh^\bC\oplus\fm^\bC\oplus\fn^\bC$ is its complexification. In addition to (\ref{eq:hrel}) and (\ref{eq:frel}) the Lie brackets are as follows (other brackets vanish),
\begin{align}
 & [\hat{\lambda},\hat{\xi}]=\hat{\xi'},
  && \text{with}\qquad \xi'=-\lambda\xi, \nonumber \\
 & [\hat{\lambda},\hat{\xi}^\fa]=\hat{\xi'}^\fa,
  && \text{with}\qquad \xi'= \lambda^\la\xi, \nonumber \\
 & [\hat{\Lambda}^\fa,\hat{\xi}]=\hat{\xi'}^\fa,
  && \text{with}\qquad \xi'=\Lambda\xi , \nonumber \\
 & [\hat{\Lambda},\hat{\xi}^\fa]=\hat{\xi'},
  && \text{with}\qquad \xi'=-\Lambda\xi , \nonumber \\
 & [\hat{\xi},\hat{\xi'}]=\hat{\Lambda},
  && \text{with}\qquad \Lambda \eta=\xi'\{\eta,\xi\}-\xi\{\eta,\xi'\}, \nonumber \\
 & [\hat{\xi}^\fa,\hat{\xi'}^\fa]=\hat{\Lambda}^\fa,
  && \text{with}\qquad \Lambda \eta=\xi\{\eta,\xi'\}-\xi'\{\eta,\xi\}, \nonumber \\
 & [\hat{\xi}^\fa,\hat{\xi'}]=\hat{\lambda},
  && \text{with}\qquad \lambda \eta=\xi'\{\xi,\eta\}.
\end{align}
We shall refer to $\fg$ as the \emph{(full) dynamical Lie algebra}. Its complexified universal enveloping algebra generates $\cA$.

In \cite{Oe:fermicoh} an ad-invariant inner product was found for $\fg^\bC$ that reduces to a (negative) multiple of the Killing form on $\fg$ in the case that $L$ is finite-dimensional. This inner product is also well defined in the present case,
\begin{multline}
 \left(\hat{\lambda_1}+\hat{\Lambda_1}^\fa+\hat{\Lambda_1'}+\hat{\xi_1}^\fa+\hat{\xi_1'},
  \hat{\lambda_2}+\hat{\Lambda_2}^\fa+\hat{\Lambda_2'}+\hat{\xi_2}^\fa+\hat{\xi_2'} \right) \\
 =2\trk(\lambda_1^\la \lambda_2)-\trk(\Lambda_2'\Lambda_1')-\trk(\Lambda_1\Lambda_2) +2\{\xi_2,\xi_1\} +2\{\xi_1',\xi_2'\}.
\label{eq:gip}
\end{multline}
However, this inner product is here not necessarily positive definite. Rather, it makes $\fg$ and $\fg^\bC$ into a Krein spaces. Only if $L$ is positive definite is $\fg$ a compact form leading to a negative definite Killing form and positive definite inner product.

\section{Coherent states}
\label{sec:cohstates}

We recall the coherent states constructed in \cite{Oe:fermicoh}. This construction generalizes to the Krein space case. The coherent states are generated by exponentiating the action of the Lie subalgebra $\fp_{-}=\fm_{-}\oplus\fn_{-}\subseteq \fg^\bC$ on the vacuum state $\vac$ of the Fock-Krein space $\cF(L)$. Explicitly, we define the map $K:\fp_-\to \cF(L)$,
\begin{equation}
 (\Lambda,\xi) \mapsto \exp\left(\hat{\Lambda}^\fa+\hat{\xi}^\fa\right)\vac .
\end{equation}

It will be useful to have explicit expressions for the coherent states $K(\Lambda,\xi)$ in terms of their decomposition by degree. We write
\begin{equation}
  K(\Lambda,\xi)=\sum_{n=0}^\infty K_n(\Lambda,\xi),
\end{equation}
where $K_n(\Lambda,\xi)\in\cF_n(L)$. Using commutation relations and formulas (\ref{eq:actcr}) and (\ref{eq:actLhat}) it is then straightforward to obtain,
\begin{align}
  & K_{2n}(\Lambda,\xi)(\eta_1,\ldots,\eta_{2n})  =\left(\frac{1}{n!}(\hat{\Lambda}^\fa)^{n}\psi_0\right)(\eta_1,\ldots,\eta_{2n}) \nonumber \\
  & \quad =\frac{1}{2^{2n}\, n!\,(2n)!} \sum_{\sigma\in S^{2n}} (-1)^{|\sigma|} \prod_{k=1}^n \{\Lambda \eta_{\sigma(2k)},\eta_{\sigma(2k-1)}\}, \label{eq:coheven} \\
  & K_{2n+1}(\Lambda,\xi)(\eta_1,\ldots,\eta_{2n+1}) =\left(\frac{n+1}{(n+1)!}\hat{\xi}(\hat{\Lambda}^\fa)^{n}\psi_0\right)(\eta_1,\ldots,\eta_{2n+1}) \nonumber \\
  & \quad =\frac{1}{2^{2n+1}\, n!\,(2n+1)!} \sum_{\sigma\in S^{2n+1}} (-1)^{|\sigma|} \{\xi,\eta_{\sigma(1)}\} \prod_{k=1}^n \{\Lambda \eta_{\sigma(2k+1)},\eta_{\sigma(2k)}\} . \label{eq:cohodd}
\end{align}

We proceed to list results from \cite{Oe:fermicoh} that generalize from the Hilbert to the Krein setting. For those results where no further remarks are made the generalization is straightforward and a proof is thus omitted.
\begin{prop}[{\cite[Proposition~7.1]{Oe:fermicoh}}]
The map $K$ is continuous, holomorphic and injective.
\end{prop}

We can evaluate the inner product of coherent states in terms of a Fredholm determinant \cite{Sim:traceideals}, generalized here to Krein spaces.

\begin{thm}[{\cite[Proposition~7.3]{Oe:fermicoh}}]
\label{thm:cohip}
Let $\Lambda,\Lambda'\in\asl(L)$ and $\xi,\xi'\in L$. Assume moreover $\lop \Lambda \Lambda'\rop<1$. Set
\begin{equation}
  b\defeq\{\xi',(\one_L-\Lambda\Lambda')^{-1}\xi\} .
  \label{eq:ipb}
\end{equation}
Then,
\begin{equation}
\langle K(\Lambda,\xi),K(\Lambda',\xi')\rangle=
 \left(1+\frac{1}{2}b\right)
\left(\det \left(\one_L-\Lambda\Lambda'\right)\right)^{\frac{1}{2}}.
\label{eq:kip}
\end{equation}
The correct branch of the square root is obtained by analytic continuation from $\Lambda=\Lambda'$.
\end{thm}

A proof of this result in the Krein setting is more involved and we shall present one
at the end of Section~\ref{sec:amplitude}. This proof is very different from the one given for the Hilbert space setting in \cite{Oe:fermicoh}. In particular, the latter does not require a restriction on the operator norm of $\Lambda \Lambda'$. This suggests that it should be possible to relax this restriction also in the present Krein space version. We shall not explore this issue further in the present paper.

\begin{prop}[{\cite[Proposition~7.4]{Oe:fermicoh}}]
\label{prop:kdense}
The image of $K$ spans a dense subspace of $\cF(L)$.
\end{prop}

Given $\psi\in\cF(L)$ define the function
\begin{equation}
f_\psi:\fp_-\to\bC\quad\text{by}\quad f_\psi(\Lambda,\xi)\defeq \langle K(\Lambda,\xi),\psi\rangle.
\label{eq:defwf}
\end{equation}
Then, $f_\psi$ is continuous and anti-holomorphic \cite[Proposition~7.5]{Oe:fermicoh}.
Denote by $\hol(\fp_-)$ the complex vector space of continuous and anti-holomorphic functions on $\fp_-$. The complex linear map
\begin{equation}
 f:\cF(L)\to\hol(\fp_-)\quad\text{given by}\quad
 \psi \mapsto f_\psi 
\end{equation}
is injective \cite[Lemma~7.6]{Oe:fermicoh}.
Let $F(\fp_-)\subseteq\hol(\fp_-)$ denote the image of $f$.

\begin{thm}[{\cite[Theorem~7.7]{Oe:fermicoh}}]
  \label{thm:repkkrein}
The complex linear isomorphism
\begin{equation}
 f:\cF(L)\to F(\fp_-) 
\end{equation}
realizes the Fock space $\cF(L)$ as a \emph{reproducing kernel Krein space} of continuous anti-holomorphic functions on the Krein space $\fp_-$ with reproducing kernel $K:\fp_-\times\fp_-\to\bC$,
\begin{equation}
 K\left((\Lambda,\xi),(\Lambda',\xi')\right)= \langle K(\Lambda,\xi),K(\Lambda',\xi')\rangle
\end{equation}
given by equation (\ref{eq:kip}) of Theorem~\ref{thm:cohip}. In particular, the reproducing property is equation (\ref{eq:defwf}).
\end{thm}

\section{The amplitude}
\label{sec:amplitude}

The physics in a time-interval $[t_1,t_2]$ can be conveniently encoded in the transition amplitudes between initial states at $t_1$ and final states at $t_2$. In quantum field theory, this is commonly calculated using the Feynman path integral. The S-matrix is an asymptotic version of this. This way of encoding physics can be generalized to spacetime regions that do not have the special form of a time-interval. Amplitudes are then linear maps from a state space associated to the boundary of the region to the complex numbers. Transition amplitudes arise as special cases for time-interval regions. It turns out that the description of fermionic systems in particular becomes richer and more interesting in this generalized setting \cite{Oe:freefermi}, also known as as \emph{general boundary quantum field theory (GBQFT)} \cite{Oe:gbqft}. A particularly striking fact is that fermionic state spaces are generally Krein spaces and not Hilbert spaces.

A quantum field theory in this context may be encoded through structures that satisfy a system of axioms. Spacetime is encoded in terms of a collection of hypersurfaces and regions, i.e., oriented submanifolds of dimensions $n-1$ and $n$. There are operations of \emph{gluing} regions and \emph{decomposing} hypersurfaces. The most important structures that determine a fermionic quantum field theory are: an assignment of an f-graded Krein space $\cH_{\Sigma}$ to each hypersurface $\Sigma$ and an assignment of an f-graded amplitude map $\rho_M:\cH_{\partial M}^\ds\to\bC$ to each region $M$ with boundary $\partial M$. (Here $\cH_{\partial M}^\ds$ is a dense subspace of $\cH_{\partial M}$.) We include the list of axioms in the appendix, but refer the reader to \cite{Oe:freefermi} for a complete explanation.

We are interested in the following in the case of a fermionic quantum field theory that arises as the quantization of a classical linear field theory, also described in \cite{Oe:freefermi}.
The classical data include a Krein space $L_{\Sigma}$ assigned to every oriented hypersurfaces $\Sigma$. This plays the role of the classical (fermionic!) phase space. Crucially, the space $L_{\overline{\Sigma}}$, associated to the hypersurface with opposite orientation, $\overline{\Sigma}$, is the same space, but with the opposite complex structure and inner product given by,
\begin{equation}
  \{\xi,\eta\}_{\overline{\Sigma}}=-\overline{\{\xi,\eta\}_{\Sigma}} .
\end{equation}
For later use we note that a map $\Lambda\in\asl(L_{\Sigma})$ is also canonically an element of $\asl(L_{\overline{\Sigma}})$. We shall implicitly use this fact in the following.

The state space $\cH_{\Sigma}$ associated to a hypersurface $\Sigma$ in the quantum theory is the fermionic Fock space $\cF(L_{\Sigma})$ of Section~\ref{sec:fockkrein}.
Recall (Axiom \textbf{(T1b)} of the appendix) that there is a map $\iota_{\Sigma}:\cF(L_{\Sigma})\to\cF(L_{\overline{\Sigma}})$ that identifies the state spaces associated to the two different orientations of the same hypersurface $\Sigma$. This map is a conjugate-linear adapted real f-graded isometry given on $\psi\in\cF_n(L_{\Sigma})$ by,
\begin{equation}
 (\iota_{\Sigma}(\psi))(\xi_1,\dots,\xi_n)= \overline{\psi(\xi_n,\dots,\xi_1)} .
\end{equation} 
Using expressions (\ref{eq:coheven}) and (\ref{eq:cohodd}) this works out for coherent states too,
\begin{equation}
  \iota_{\Sigma}(K(\Lambda,\xi))=K(\Lambda,-\xi) .
  \label{eq:iotacoh}
\end{equation}

Let $\Sigma=\Sigma_1\cup\Sigma_2$ be a hypersurface decomposition. The classical data then satisfy $L_{\Sigma}=L_{\Sigma_1}\oplus L_{\Sigma_2}$. Recall (Axiom \textbf{(T2)} of the appendix) that associated to this is an isometric isomorphism of f-graded Krein spaces, $\tau_{\Sigma_1,\Sigma_2;\Sigma}:\cF(L_{\Sigma_1})\ctprod\cF(L_{\Sigma_2})\to\cF(L_\Sigma)$. (Here, $\ctprod$ denotes the completed tensor product.) For $\psi\in\cF_m(L_{\Sigma_1})$ and $\psi_2\in\cF_n(L_{\Sigma_2})$ this is given by,
\begin{multline}
 \left(\tau_{\Sigma_1,\Sigma_2;\Sigma}(\psi_1\tprod\psi_2)\right)\left((\eta_1,\xi_1),\dots,(\eta_{m+n},\xi_{m+n})\right)\\
= \frac{1}{(m+n)!}\sum_{\sigma\in S^{m+n}}(-1)^{|\sigma|}
 \psi_1(\eta_{\sigma(1)},\dots,\eta_{\sigma(m)})\psi_2(\xi_{\sigma(m+1)},\dots,\xi_{\sigma(m+n)}) .
\end{multline}
For coherent states this map takes the following form, as can be verified straightforwardly,
\begin{align}
    & \tau_{\Sigma_1,\Sigma_2;\Sigma}(K(\Lambda,\xi)\tprod K(\Lambda',\xi'))
  = K(\Lambda+\Lambda'+\tilde{\Lambda},\xi+\xi'), \label{eq:taucoh} \\
    & \quad\text{where}\quad
    \tilde{\Lambda}(\eta)\defeq\frac{1}{2}\left(\xi\{\eta,\xi'\}-\xi'\{\eta,\xi\}\right) .
\end{align}
Note that $\Lambda$ is extended here implicitly from an element in $\asl(L_{\Sigma_1})$ to en element in $\asl(L_{\Sigma})$ and correspondingly for $\Lambda'$.

We now come to the object of principal interest of this section, the quantum amplitude associated to a spacetime region $M$. The classical dynamics in $M$ is encoded in a conjugate-linear involutive adapted real anti-isometry $u:L_{\partial M}\to L_{\partial M}$ on the boundary phase space $L_{\partial M}$. (This map is denoted $u_M$ in \cite{Oe:freefermi}.)
The amplitude map $\rho_M:\cF(L_{\partial M})\to\bC$ (see Axiom \textbf{(T4)} of the appendix) is then given as follows.
For a state of odd degree the amplitude map vanishes,
\begin{equation}
 \rho_M(\psi)= 0 \qquad\text{if}\quad \psi\in \cF_{2n+1}(L_{\partial M}) .
\end{equation}
Since we view the target space $\bC$ of $\rho_M$ as of even f-degree, this is equivalent to saying that $\rho_M$ is f-graded. Moreover, for the vacuum state $\psi_0\in \cF(L_{\partial M})$ the amplitude is the unit,
\begin{equation}
  \rho_M(\vac)=1 .
\end{equation}
For a state of fixed even degree $\psi\in\cF_{2n}(L_{\partial M})$ the amplitude is given by,
\begin{equation}
 \rho_M(\psi)= \frac{(2n)!}{n!}
 \sum_{j_1,\dots,j_n\in I} \{\zeta_{j_1},\zeta_{j_1},\} \cdots \{\zeta_{j_n},\zeta_{j_n}\}
 \psi(u \zeta_{j_1},\zeta_{j_1},\dots,u \zeta_{j_n},\zeta_{j_n}) .
 \label{eq:amplefock}
\end{equation}
(Compare formula (51) in \cite{Oe:locqft}. Again, a relative factor arises from the difference between using a real and a complex basis.)

In order to evaluate the amplitude $\rho_M$ on a coherent state $K(\Lambda,\xi)\in\cF(L_{\partial M})$ we begin by considering its value on the component $K_{2n}(\Lambda,\xi)\in\cF_{2n}(L_{\partial M})$ of fixed even degree. With the formulas (\ref{eq:coheven}) and (\ref{eq:amplefock}) we find that this can be expressed as follows,
\begin{multline}
  \rho_M(K_{2n}(\Lambda,\xi)) = \frac{1}{2^{2n}(n!)^2} R, \quad \text{with}\quad R \defeq \sum_{\sigma\in S^{2n}} R_{\sigma}, \quad \text{where}, \\
  R_\sigma \defeq  (-1)^{|\sigma|}
   \sum_{j_1,\dots,j_n\in I} \{\zeta_{j_1},\zeta_{j_1},\} \cdots \{\zeta_{j_n},\zeta_{j_n}\}
  \prod_{k=1}^n \{\Lambda\eta_{\sigma(2k)},\eta_{\sigma(2k-1)}\}_{\partial M},  \label{eq:amplcoheprob}
\end{multline}
and where we introduce the definitions, $\eta_{2k} \defeq\zeta_k$ and $\eta_{2k-1}\defeq u\zeta_k$. A closer look at expression $R_{\sigma}$ reveals that it factorizes as follows. We identify the variables $\eta_i$ with vertices $i\in\{1,\ldots,2n\}$ of a graph. For all $k$ connect with an edge the vertex $2k$ with the vertex $2k-1$, corresponding to the dependency on the same basis element $\zeta_k$. Also connect for any $k$ the vertex $\sigma(2k)$ with the vertex $\sigma(2k-1)$. This corresponds to common appearance in the inner product $\{\Lambda\eta_{\sigma(2k)},\eta_{\sigma(2k-1)}\}$. It is then easy to see that $R_{\sigma}$ factors into one component for each connected component of the resulting graph. What is more, using extensively the conjugate anti-symmetry property (\ref{eq:conjas}) both of $u$ and of $\Lambda$, it turns out that any factor involving $2k$ variables $\eta_i$ can be brought into the same form given by,
\begin{equation}
  f_k\defeq -\trk\left((u\Lambda)^k\right) .
  \label{eq:fk}
\end{equation}
What is more, any $R_{\sigma}$ is precisely a product $f_1^{j_1}\cdots f_n^{j_n}$ of these factors. In order to work out the precise expressions it is useful to analyze the underlying combinatorial problem more abstractly.

Consider the following combinatorial problem. Let $n\in\N$ and recall that $S^{2n}$ denotes the symmetric group in $2n$ elements. We associate to $\sigma\in S^{2n}$ a monomial $p_{\sigma}$ in variables $x_1,\ldots,x_n$ as follows. Consider a graph with $2n$ vertices labeled $1$ to $2n$. For each $k\in \{1,\ldots,n\}$ connect the vertices $2k-1$ and $2k$ with an edge and also connect the vertices $\sigma(2k-1)$ and $\sigma(2k)$ with an edge. The resulting graph has $2n$ vertices and $2n$ edges. (Note that we allow multiple edges between vertices.) This graph decomposes into a disjoint union of cyclic graphs with even edge number. Denote the multiplicity of the cyclic graph with $2k$ edges in this decomposition by $j_k$. Then,
\begin{equation}
  p_{\sigma}\defeq x_1^{j_1}\cdots x_n^{j_n} .
\end{equation}
Moreover, we define the polynomial $p_n$ as the sum over these monomials for all permutations,
\begin{equation}
  p_n\defeq \sum_{\sigma\in S^{2n}} p_{\sigma} .
  \label{eq:psumperm}
\end{equation}
We also define $p_0\defeq 1$.

\begin{lem}
  The polynomials $p_n$ satisfy the following recursion relations,
  \begin{equation}
    p_n=\frac{1}{2n}\sum_{k=1}^n 2^{2k} \left(\frac{n!}{(n-k)!}\right)^2 x_k\, p_{n-k}
    \quad\forall n\in\N .
    \label{eq:precrel}
  \end{equation}
\end{lem}
\begin{proof}
  Fix $n\in\N$ and $\sigma\in S^{2n}$. As a first observation, we note that $p_\sigma$ remains unchanged if for any $k$ we interchange the label of the vertex $2k-1$ with that of the vertex $2k$. This generates an invariance under the action of a $k$-fold product of $S^2$. Similarly, $p_\sigma$ remains unchanged if we interchange for $k\neq m$ the label of the vertex $2k-1$ with that of the vertex $2m-1$ and at the same time that of $2k$ with that of $2m$. This generates an invariance under an action of $S^n$. We obtain another instance of the same type of invariances by replacing everywhere the label $k$ with the label $\sigma(k)$. That is, in total we obtain an invariance under an action of $(S^2)^{2n}\times (S^n)^2$.
  
  We now consider the sum (\ref{eq:psumperm}). Due to the mentioned invariances we can restrict the permutations to satisfy $\sigma(1)=1$ without altering the sum, extracting a factor $2n$ for the corresponding multiplicity,
\begin{equation}
  p_n= 2n \sum_{\substack{\sigma\in S^{2n}\\ \sigma(1)=1}} p_\sigma .
\end{equation}
We now split the remaining sum into two parts, the first with the terms where $\sigma(2)=2$ and the second for $\sigma(2)>2$. (Note that $\sigma(2)=1$ cannot occur since $\sigma(1)=1$.) In terms of graphs, the first term yields a digon graph with the vertices labeled $1$ and $2$ and a sum over graphs for the remaining vertices indexed by $S^{2(n-1)}$ acting on these remaining vertices. That is, we get a factor $x_1$ for the digon graph and the polynomial $p_{n-1}$ for the remainder,
\begin{equation}
  p_n= 2n \left( x_1\, p_{n-1} + \sum_{\substack{\sigma\in S^{2n}\\ \sigma(1)=1,\sigma(2)>2}} p_\sigma\right) .
\end{equation}
In the second term we can fix $\sigma(2)=3$, again due to the mentioned symmetries, extracting a multiplicity of $2(n-1)$. On top of that we can fix $\sigma(3)=2$, yielding another factor $2(n-1)$,
\begin{equation}
  p_n= 2n \left( x_1\, p_{n-1} + (2(n-1))^2 \sum_{\substack{\sigma\in S^{2n}\\ \sigma(1)=1,\sigma(2)=3,\sigma(3)=2}} p_\sigma \right) .
\end{equation}
Next, we split the sum into a part with $\sigma(4)=4$ and another part with $\sigma(4)>4$. In the first
part in terms of graphs we get a 4-gon (square graph) made out of the first four vertices, yielding $x_2$, and a remaining sum indexed by $S^{2(n-2)}$, yielding $p_{n-2}$,
\begin{equation}
  p_n= 2n \left( x_1\, p_{n-1} + (2(n-1))^2 \left( x_2\, p_{n-2} + \sum_{\substack{\sigma\in S^{2n}\\ \sigma(1)=1,\sigma(2)=3,\sigma(3)=2,\sigma(4)>4}} p_\sigma \right)\right) .
\end{equation}
Iterating this process we arrive at the recursion relation,
\begin{equation}
  p_n= 2n \left( x_1\, p_{n-1} + (2(n-1))^2 \left( x_2\, p_{n-2} + \ldots (2\cdot 2)^2\left(x_{n-1} p_1 + (2\cdot 1)^2\, x_n p_0\right)\ldots\right)\right) .
\end{equation}
This can be conveniently rewritten as (\ref{eq:precrel}). \myqed
\end{proof}

We introduce the polynomials $q_n$ as rescaling of the $p_n$,\footnote{Note that $2^{2n} (n!)^2$ is precisely the order of the symmetry group $(S^2)^{2n}\times (S^n)^2$ mentioned in the proof.}
\begin{equation}
  q_n\defeq \frac{p_n}{2^{2n} (n!)^2} .
\end{equation}
Also, we introduce rescalings $y_k$ of the variables $x_k$,
\begin{equation}
  y_k\defeq\frac{1}{2} x_k .
\end{equation}
The recursion relation (\ref{eq:precrel}) then simplifies considerably,
\begin{equation}
  q_n=\frac{1}{n}\sum_{k=1}^n y_k\, q_{n-k}\quad\forall n\in\N .
\end{equation}
It turns out that precisely this recurrence relation characterizes the \emph{cycle index} of the symmetric group in $n$ elements, a concept introduced by Pólya for combinatorial problems similar to the type we are considering \cite{Pol:Kombanz}. As noted by Pólya this cycle index admits the following explicit expression,
\begin{equation}
  q_n=\sum_{j_1+2 j_2+\ldots+n j_n=n} \prod_{k=1}^n \frac{1}{j_k!}  \left(\frac{y_k}{k}\right)^{j_k} .
\end{equation}

This solves precisely our problem concerning the amplitude of $K_{2n}(\Lambda,\xi)$ if we evaluate at $x_k=f_k$. Then, $p_n$ becomes $R$ in expression (\ref{eq:amplcoheprob}), while $q_n$ becomes the amplitude itself.

\begin{lem}
  Let $M$ be a region, $\Lambda\in\asl(L_{\partial M})$ and $\xi\in L_{\partial M}$. Assume moreover that $u\Lambda$ is trace class. Then,
  \begin{equation}
    \rho_M(K_{2n}(\Lambda,\xi))=\sum_{j_1+2 j_2+\ldots+n j_n=n} \frac{1}{j_k!}  \left(\frac{f_k}{2 k}\right)^{j_k},\quad\text{where $f_k$ is defined in (\ref{eq:fk}}).
  \end{equation}
\end{lem}

Our next task will be to obtain the complete amplitude of the coherent state by summing over all even Fock degrees,
\begin{equation}
  \rho_M(K(\Lambda,\xi))=\sum_{n=0}^\infty \rho_M(K_{2n}(\Lambda,\xi)).
\end{equation}
Again, it will be fruitful to consider the problem first in a more abstract setting.

By a \emph{formal power series} in variables $y_1,y_2,\ldots$ we mean an assignment of a real coefficient to each finite monomial that can be formed with the variables $y_1,y_2,\ldots$. We denote the ring of these formal power series by $\R[[y_1,y_2,\ldots]]$.

\begin{lem}
  In $\R[[y_1,y_2,\ldots]]$ we have the equality,
  \begin{equation}
    \sum_{n=0}^\infty q_n=\exp\left(\sum_{k=1}^\infty \frac{y_k}{k}\right) .
  \end{equation}
\end{lem}
\begin{proof}
  \begin{align}
  \sum_{n=0}^\infty q_n
  & =\sum_{n=0}^\infty\, \sum_{j_1+2 j_2+\ldots+n j_n=n}\, \prod_{k=1}^n \frac{1}{j_k!}\left(\frac{y_k}{k}\right)^{j_k} \\
  & =\sum_{j_1,j_2,\ldots=0}^\infty\, \prod_{k=1}^\infty \frac{1}{j_k!}\left(\frac{y_k}{k}\right)^{j_k}
  \label{eq:sumq2} \\
  & =\prod_{k=1}^\infty \sum_{j=0}^\infty \frac{1}{j!}\left(\frac{y_k}{k}\right)^{j} \\
  & =\prod_{k=1}^\infty \exp\left(\frac{y_k}{k}\right) \\
  & =\exp\left(\sum_{k=1}^\infty \frac{y_k}{k}\right).
  \end{align}
  Note that working in formal power series means that all expressions are interpreted as (infinite) sums of finite monomials with coefficients. In particular, the sum in expression (\ref{eq:sumq2}) is over assignments of non-negative integers to indices $j_1,j_2,\ldots$ in such a way that only finitely many indices are non-zero. This same expression also shows that the coefficient of each monomial is well defined. \myqed
\end{proof}

\begin{thm}
  \label{thm:amplcoh}
  Let $M$ be a region, $\Lambda\in\asl(L_{\partial M})$ and $\xi\in L_{\partial M}$. Assume moreover that $u\Lambda$ is trace class and $\lop u \Lambda\rop<1$. Then,
  \begin{equation}
    \rho_M(K(\Lambda,\xi)) = \left(\det\left(\one-u\Lambda\right)\right)^\frac{1}{2} .
  \end{equation}
\end{thm}
\begin{proof}
  With the previous lemma we get,
\begin{align}
  \rho_M(K(\Lambda,\xi)) & =\sum_{n=0}^\infty \rho_M(K_{2n}(\Lambda,\xi)) \\
  & =\exp\left(\sum_{k=1}^\infty -\frac{1}{2k}\trk\left((u\Lambda)^k\right)\right) \label{eq:amplcohtr}\\
  & =\exp\left(\frac{1}{2}\trk\left(\sum_{k=1}^\infty -\frac{1}{k}(u\Lambda)^k\right)\right) \\
  & =\exp\left(\frac{1}{2}\trk\left(\ln\left(\one-u\Lambda\right)\right)\right) \\
  & =\left(\exp\left(\trk\left(\ln\left(\one-u\Lambda\right)\right)\right)\right)^{\frac{1}{2}} \\
  & =\left(\det\left(\one-u\Lambda\right)\right)^\frac{1}{2} .
\end{align}
Here we use again the Fredholm determinant for Krein spaces. \myqed
\end{proof}
Note that the condition on the operator norm of $u\Lambda$ serves to guarantee convergence of the sums in the proof. However, it is likely that this condition can be relaxed by analogy to Proposition~7.3 of \cite{Oe:fermicoh}. Recall corresponding comments after Theorem~\ref{thm:cohip} in Section~\ref{sec:cohstates}.

Finally, we recall that the inner product on the state spaces and the amplitude are intimately related due to Axiom \textbf{(T3x)} of the appendix. In particular, this means that we can recover the inner product on the state space $\cH_{\Sigma}$ for the hypersurface $\Sigma$ as a special case of the amplitude. To this end we need to evaluate the latter on the \emph{slice region} $\hat{\Sigma}$ obtained by infinitesimally thickening $\Sigma$.
Note here that $\partial\hat{\Sigma}=\overline{\Sigma}\cup\Sigma$. Correspondingly, we have $L_{\partial \hat{\Sigma}}=L_{\overline{\Sigma}}\oplus L_{\Sigma}$. For the details, see \cite{Oe:freefermi}.
Concretely, the inner product on $\Sigma$ is given in terms of the amplitude on $\hat{\Sigma}$ as,
\begin{equation}
\langle \psi',\psi\rangle_{\Sigma}=\rho_{\hat{\Sigma}}\circ\tau_{\overline{\Sigma},\Sigma;\partial \hat{\Sigma}}(\iota_{\Sigma}(\psi')\tprod\psi) .
\end{equation}
For coherent states we get, using (\ref{eq:iotacoh}) and (\ref{eq:taucoh}),
\begin{multline}
  S\defeq \langle K(\Lambda,\xi),K(\Lambda',\xi')\rangle_{\Sigma}=\rho_{\hat{\Sigma}}\left(K(\Lambda,-\xi)\tprod K(\Lambda',\xi')\right)\\
  = \rho_{\hat{\Sigma}}(K(\Lambda_{\overline{\Sigma}}+\Lambda'_{\Sigma}+\tilde{\Lambda},-\xi_{\overline{\Sigma}}+\xi'_{\Sigma})) .
  \label{eq:ipamplcoh}
\end{multline}
Here $\tilde{\Lambda}(\eta_{\overline{\Sigma}}+\eta_{\Sigma}')=\frac{1}{2}\left(\xi'_{\Sigma}\{\eta,\xi\}_{\overline{\Sigma}}-\xi_{\overline{\Sigma}}\{\eta,\xi'\}_{\Sigma}\right)$. We use the subscripts $\Sigma$ and $\Sigma'$ to indicate in which component of $L_{\partial \hat{\Sigma}}=L_{\overline{\Sigma}}\oplus L_{\Sigma}$ a certain object lives. With Theorem~\ref{thm:amplcoh} we can now evaluate in principle the right hand side of expression (\ref{eq:ipamplcoh}). Note that for a slice region, $u$ (denoted $u_{\hat{\Sigma}}$ in \cite{Oe:freefermi}) takes the particularly simple form,
\begin{equation}
  u(\eta_{\overline{\Sigma}}+\eta_{\Sigma}')=\eta'_{\overline{\Sigma}}+\eta_{\Sigma} .
\end{equation}
That is, $u$ simply interchanges the components of $L_{\partial \hat{\Sigma}}=L_{\overline{\Sigma}}\oplus L_{\Sigma}$. We obtain,
\begin{equation}
  \rho_{\hat{\Sigma}}(K(\Lambda_{\overline{\Sigma}}+\Lambda'_{\Sigma}+\tilde{\Lambda},-\xi_{\overline{\Sigma}}+\xi'_{\Sigma}))={\det}_{\partial\hat{\Sigma}}\left(\one-u(\Lambda_{\overline{\Sigma}}+\Lambda'_{\Sigma}+\tilde{\Lambda})\right)^{\frac{1}{2}} .
\end{equation}
The difficulty lies in obtaining a more explicit expression for the right hand side. In fact, it can be shown that this evaluates precisely to the expression given in Theorem~\ref{thm:cohip}.

\begin{proof}[Proof of Theorem~\ref{thm:cohip}]
  Instead of the determinant formula we start with formula (\ref{eq:amplcohtr}). We write this as,
  \begin{equation}
    S=\exp\left(\sum_{k=1}^\infty -\frac{1}{2k}\trk_{\partial\hat{\Sigma}}\left((u\underline{\Lambda}+u\tilde{\Lambda})^k\right)\right), \qquad\text{where}\qquad \underline{\Lambda}\defeq\Lambda_{\overline{\Sigma}}+\Lambda'_{\Sigma} .
    \label{eq:iptr1}
  \end{equation}
  The first step will be to expand the innermost sum and reorganize terms in terms of powers of $u\tilde{\Lambda}$. A useful observation is that terms of the form $\trk_{\partial\hat{\Sigma}}\left((u\underline{\Lambda})^k\right)$ vanish if $k$ is odd, since $u\underline{\Lambda}$ exchanges the components $L_{\overline{\Sigma}}$ and $L_{\Sigma}$. But an odd number of exchanges will lead to a vanishing trace. Another observation is that any product of the form $u\tilde{\Lambda}(u\Lambda)^k u\tilde{\Lambda}$ vanishes if $k$ is odd. This is due to conjugate-linearity of $\Lambda$ and $\Lambda'$ and the symmetric appearance of the vector $\xi$ or $\xi'$ respectively in those expressions, as can be seen by inspection. Thus a term that contains $n$ powers of $u\tilde{\Lambda}$ can be brought into the form
  \begin{equation}
    \trk_{\partial\hat{\Sigma}}\left(u\tilde{\Lambda}(u\underline{\Lambda})^{2k_1} u\tilde{\Lambda}(u\underline{\Lambda})^{2k_2}\cdots u\tilde{\Lambda}(u\underline{\Lambda})^{2k_n}\right) ,
  \end{equation}
  by using cyclic permutation symmetry of the trace. For fixed $n$ the relative multiplicity of the term with powers $k1_,\ldots,k_n$ if we sum independently over $k_1,\ldots,k_n$ is given by
  \begin{equation}
  \frac{2k_1+\cdots+2k_n+n}{n} .
  \end{equation}
  Combining with the weight factors in (\ref{eq:iptr1}) we can restructure the expression as,
  \begin{multline}
    S=\exp\left(\sum_{k=1}^\infty -\frac{1}{4k}\trk_{\partial\hat{\Sigma}}
    \left((u\underline{\Lambda})^{2k}\right)\right) \\
    \exp\left(\sum_{n=1}^\infty -\frac{1}{2n}\sum_{k_1,\ldots,k_n =0}^\infty\trk_{\partial\hat{\Sigma}}
    \left(u\tilde{\Lambda}(u\underline{\Lambda})^{2k_1}\cdots u\tilde{\Lambda}(u\underline{\Lambda})^{2k_n}\right)\right) .
  \label{eq:iptr2}
  \end{multline}
  As for the first factor we observe that the trace decomposes into a sum of traces for each of the components $L_{\overline{\Sigma}}$ and $L_{\Sigma}$ of $L_{\partial\hat{\Sigma}}$ with equal value. More precisely, we find,
  \begin{equation}
    \trk_{\partial\hat{\Sigma}} \left((u\underline{\Lambda})^{2k}\right)
    =2\, \trk_{\Sigma}\left((\Lambda \Lambda')^k\right) .
  \end{equation}
  As for the second factor we observe the same decomposition of the trace into a sum of two equal traces. What is more, these component traces factorize into $n$ factors as follows,
  \begin{equation}
    \trk_{\partial\hat{\Sigma}} \left(u\tilde{\Lambda}(u\underline{\Lambda})^{2k_1}
    \cdots u\tilde{\Lambda}(u\underline{\Lambda})^{2k_n}\right)
    =2 g_{k_1} \cdots g_{k_n},\;\text{with}\;
    g_k = - \frac{1}{2}\{\xi',(\Lambda \Lambda')^k \xi\}_{\Sigma} .
  \end{equation}
  This yields,
  \begin{align}
    S & =\exp\left(\sum_{k=1}^\infty -\frac{1}{2k}
    \trk_{\Sigma}\left((\Lambda \Lambda')^k\right)\right)
    \exp\left(\sum_{n=1}^\infty -\frac{1}{n}\sum_{k_1,\ldots,k_n =0}^\infty
    g_{k_1} \cdots g_{k_n}\right) \\
     & =\exp\left(\frac{1}{2}\trk_{\Sigma}\left(\sum_{k=1}^\infty -\frac{1}{k}
    (\Lambda \Lambda')^k\right)\right)
    \exp\left(\sum_{n=1}^\infty -\frac{1}{n}\left(\sum_{k=0}^\infty
    g_{k}\right)^n\right) .
  \end{align}
  For the sum over the factors $g_k$ we get,
  \begin{multline}
    \sum_{k=0}^\infty g_{k}= \sum_{k=0}^\infty - \frac{1}{2}\{\xi',(\Lambda \Lambda')^k \xi\}_{\Sigma} = - \frac{1}{2}\left\{\xi',\left(\sum_{k=0}^\infty (\Lambda \Lambda')^k\right) \xi\right\}_{\Sigma} \\
    = - \frac{1}{2}\left\{\xi',\left(\one-\Lambda \Lambda'\right)^{-1} \xi\right\}_{\Sigma}
    = -\frac{1}{2} b ,
  \end{multline}
  where $b$ is defined as in equation (\ref{eq:ipb}).
  (Note that convergence here requires $\lop\Lambda\Lambda'\rop<1$.) Further we get,
  \begin{align}
    S & =\exp\left(\frac{1}{2}
    \trk_{\Sigma}\left(\ln\left(\one-\Lambda \Lambda'\right)\right)\right)
    \exp\left(\ln \left(1+\frac{1}{2}b\right)\right) \\
    & =\left(\exp\left(\trk_{\Sigma}\left(\ln\left(\one-\Lambda \Lambda'\right)\right)\right)
    \right)^{\frac{1}{2}}
    \left(1+\frac{1}{2}b\right) \\
    & =\left(\det\left(\one-\Lambda \Lambda'\right)\right)^{\frac{1}{2}}
    \left(1+\frac{1}{2}b\right) .
  \end{align}
  This reproduces expression (\ref{eq:kip}), concluding the proof. \myqed
\end{proof}

\section*{Appendix: GBQFT axioms}

The following is a version of the axiomatic system of general boundary quantum field theory from \cite{Oe:freefermi}, slightly modified.

\begin{description}
\item[\textbf{(T1)}] Associated to each hypersurface $\Sigma$ is a complex
  separable f-graded Krein space $\cH_\Sigma$, called the \emph{state space} of
  $\Sigma$. We denote its indefinite inner product by
  $\langle\cdot,\cdot\rangle_\Sigma$.
\item[\textbf{(T1b)}] Associated to each hypersurface $\Sigma$ is a conjugate linear adapted f-graded isometry $\iota_\Sigma:\cH_\Sigma\to\cH_{\overline{\Sigma}}$. This map is an involution in the sense that $\iota_{\overline{\Sigma}}\circ\iota_\Sigma$ is the identity on  $\cH_\Sigma$.
\item[\textbf{(T2)}] Suppose the hypersurface $\Sigma$ decomposes into a union of hypersurfaces $\Sigma=\Sigma_1\cup\cdots\cup\Sigma_n$. Then,
  there is an isometric isomorphism of Krein spaces
  $\tau_{\Sigma_1,\dots,\Sigma_n;\Sigma}:\cH_{\Sigma_1}\ctprod\cdots\ctprod\cH_{\Sigma_n}\to\cH_\Sigma$. The maps $\tau$ satisfy obvious associativity conditions. Moreover, in the case $n=2$ the map $\tau_{\Sigma_2,\Sigma_1;\Sigma}^{-1}\circ \tau_{\Sigma_1,\Sigma_2;\Sigma}:\cH_{\Sigma_1}\ctprod\cH_{\Sigma_2}\to\cH_{\Sigma_2}\ctprod\cH_{\Sigma_1}$ is the f-graded transposition,
\begin{equation}
 \psi_1\tprod\psi_2\mapsto (-1)^{\fdg{\psi_1}\cdot\fdg{\psi_2}}\psi_2\tprod\psi_1 .
\end{equation}
\item[\textbf{(T2b)}] Orientation change and decomposition are compatible in an f-graded sense. That is, for a decomposition of hypersurfaces $\Sigma=\Sigma_1\cup\Sigma_2$ we have
\begin{equation}
\tau_{\overline{\Sigma}_1,\overline{\Sigma}_2;\overline{\Sigma}}
 \left((\iota_{\Sigma_1}\tprod\iota_{\Sigma_2})(\psi_1\tprod\psi_2)\right)
  =(-1)^{\fdg{\psi_1}\cdot\fdg{\psi_2}}\iota_\Sigma\left(\tau_{\Sigma_1,\Sigma_2;\Sigma}(\psi_1\tprod\psi_2)\right) .
\end{equation}
\item[\textbf{(T4)}] Associated with each region $M$ is an f-graded linear map
  from a dense subspace $\cH_{\partial M}^\ds$ of the state space
  $\cH_{\partial M}$ of its boundary $\partial M$ (which carries the
  induced orientation) to the complex
  numbers, $\rho_M:\cH_{\partial M}^\ds\to\bC$. This is called the
  \emph{amplitude} map.
\item[\textbf{(T3x)}] Let $\Sigma$ be a hypersurface. The boundary $\partial\hat{\Sigma}$ of the associated slice region $\hat{\Sigma}$ decomposes into the disjoint union $\partial\hat{\Sigma}=\overline{\Sigma}\sqcup\Sigma'$, where $\Sigma'$ denotes a second copy of $\Sigma$. Then, $\rho_{\hat{\Sigma}}$ is well defined on $\tau_{\overline{\Sigma},\Sigma';\partial\hat{\Sigma}}(\cH_{\overline{\Sigma}}\tprod\cH_{\Sigma'})\subseteq\cH_{\partial\hat{\Sigma}}$. Moreover, $\rho_{\hat{\Sigma}}\circ\tau_{\overline{\Sigma},\Sigma';\partial\hat{\Sigma}}$ restricts to a bilinear pairing $(\cdot,\cdot)_\Sigma:\cH_{\overline{\Sigma}}\times\cH_{\Sigma'}\to\bC$ such that $\langle\cdot,\cdot\rangle_\Sigma=(\iota_\Sigma(\cdot),\cdot)_\Sigma$.
\item[\textbf{(T5a)}] Let $M_1$ and $M_2$ be regions and $M\defeq M_1\sqcup M_2$ be their disjoint union. Then $\partial M=\partial M_1\sqcup \partial M_2$ is also a disjoint union and $\tau_{\partial M_1,\partial M_2;\partial M}(\cH_{\partial M_1}^\ds\tprod \cH_{\partial M_2}^\ds)\subseteq \cH_{\partial M}^\ds$. Moreover, for all $\psi_1\in\cH_{\partial M_1}^\ds$ and $\psi_2\in\cH_{\partial M_2}^\ds$,
\begin{equation}
 \rho_{M}\left(\tau_{\partial M_1,\partial M_2;\partial M}(\psi_1\tprod\psi_2)\right)= \rho_{M_1}(\psi_1)\rho_{M_2}(\psi_2) .
\end{equation}
\item[\textbf{(T5b)}] Let $M$ be a region with its boundary decomposing as a union $\partial M=\Sigma_1\cup\Sigma\cup \overline{\Sigma'}$, where $\Sigma'$ is a copy of $\Sigma$. Let $M_1$ denote the gluing of $M$ with itself along $\Sigma,\overline{\Sigma'}$ and suppose that $M_1$ is a region. Then, $\tau_{\Sigma_1,\Sigma,\overline{\Sigma'};\partial M}(\psi\tprod\xi\tprod\iota_\Sigma(\xi))\in\cH_{\partial M}^\ds$ for all $\psi\in\cH_{\partial M_1}^\ds$ and $\xi\in\cH_\Sigma$. Moreover, for any ON-basis $\{\zeta_i\}_{i\in I}$ of $\cH_\Sigma$, we have for all $\psi\in\cH_{\partial M_1}^\ds$,
\begin{equation}
 \rho_{M_1}(\psi)\cdot c(M;\Sigma,\overline{\Sigma'})
 =\sum_{i\in I}(-1)^{\sig{\zeta_i}}\rho_M\left(\tau_{\Sigma_1,\Sigma,\overline{\Sigma'};\partial M}(\psi\tprod\zeta_i\tprod\iota_\Sigma(\zeta_i))\right),
\label{eq:glueax1}
\end{equation}
where $c(M;\Sigma,\overline{\Sigma'})\in\bC\setminus\{0\}$ is called the \emph{gluing anomaly factor} and depends only on the geometric data.
\end{description}

\acknowledgement{This work was partially supported by CONACYT project grant 259258 and UNAM-DGAPA-PAPIIT project grant IN109415.}

\bibliographystyle{stdnodoi} 
\bibliography{stdrefsb}
\end{document}